\documentclass[sigconf]{acmart}

\AtBeginDocument{%
  \providecommand\BibTeX{{%
    \normalfont B\kern-0.5em{\scshape i\kern-0.25em b}\kern-0.8em\TeX}}}

\setcopyright{acmlicensed}

\copyrightyear{2026}
\acmYear{2026}
\setcopyright{cc}
\setcctype{by}
\acmConference[RecSys '26]{20th ACM Conference on Recommender Systems}{September 27-October 02, 2026}{Minneapolis, MN, USA}
\acmBooktitle{20th ACM Conference on Recommender Systems (RecSys '26), September 27-October 02, 2026, Minneapolis, MN, USA}
\acmDOI{10.1145/3773078.3831892}
\acmISBN{979-8-4007-2284-4/2026/09}

\usepackage[ruled,vlined,linesnumbered,commentsnumbered]{algorithm2e}
\usepackage{amsthm}
\usepackage{subcaption}

\usepackage{enumitem}

\newtheorem{theorem}{Theorem}

\theoremstyle{definition}
\newtheorem{definition}[theorem]{Definition}

\begin{document}


\title[Embedding Subspace Partitioning]{%
  Embedding Subspace Partitioning for Dynamic%
  \texorpdfstring{\\}{ }%
  Multi-Objective Retrieval%
}

\author{Shaobo Zhang}
\orcid{0009-0004-1725-416X}
\affiliation{%
  \institution{LinkedIn}
  \city{Mountain View}
  \state{CA}
  \country{USA}
}
\email{shaozhang@linkedin.com}

\author{Alice Leung}
\orcid{0009-0000-4931-9977}
\affiliation{%
  \institution{LinkedIn}
  \city{Mountain View}
  \state{CA}
  \country{USA}
}
\email{alleung@linkedin.com}

\author{Yunxiang Ren}
\orcid{0000-0003-2158-9451}
\affiliation{%
  \institution{LinkedIn}
  \city{Mountain View}
  \state{CA}
  \country{USA}
}
\email{yunren@linkedin.com}

\author{Ping Liu}
\orcid{0000-0002-0866-8801}
\affiliation{%
  \institution{LinkedIn}
  \city{Mountain View}
  \state{CA}
  \country{USA}
}
\email{piliu@linkedin.com}

\author{Yuchin Juan}
\orcid{0009-0006-0956-0575}
\affiliation{%
  \institution{LinkedIn}
  \city{Mountain View}
  \state{CA}
  \country{USA}
}
\email{yjuan@linkedin.com}

\author{Qianqi Shen}
\orcid{0000-0002-9323-6404}
\affiliation{%
  \institution{LinkedIn}
  \city{Mountain View}
  \state{CA}
  \country{USA}
}
\email{qishen@linkedin.com}

\author{Benjamin Le}
\orcid{0009-0001-6096-478X}
\affiliation{%
  \institution{LinkedIn}
  \city{Mountain View}
  \state{CA}
  \country{USA}
}
\email{ble@linkedin.com}

\author{Jianqiang Shen}
\orcid{0009-0000-5258-6523}
\affiliation{%
  \institution{LinkedIn}
  \city{Mountain View}
  \state{CA}
  \country{USA}
}
\email{jshen@linkedin.com}

\author{Chengming Jiang}
\orcid{0009-0009-0115-4603}
\affiliation{%
  \institution{LinkedIn}
  \city{Mountain View}
  \state{CA}
  \country{USA}
}
\email{cjiang@linkedin.com}

\author{Ko-Cheng Wang}
\orcid{0009-0004-4684-5620}
\affiliation{%
  \institution{LinkedIn}
  \city{Mountain View}
  \state{CA}
  \country{USA}
}
\email{lucwang@linkedin.com}

\author{Vidya Krishnamurthy}
\orcid{0009-0007-8725-7923}
\affiliation{%
  \institution{LinkedIn}
  \city{Mountain View}
  \state{CA}
  \country{USA}
}
\email{srkrishnamurthy@linkedin.com}

\author{Caleb Johnson}
\orcid{0009-0000-9735-3328}
\affiliation{%
  \institution{LinkedIn}
  \city{Mountain View}
  \state{CA}
  \country{USA}
}
\email{\ \ cajohnson@linkedin.com}

\author{Fedor Borisyuk}
\orcid{0009-0005-8171-7656}
\affiliation{%
  \institution{LinkedIn}
  \city{Mountain View}
  \state{CA}
  \country{USA}
}
\email{fborisyuk@linkedin.com}

\author{Luke Simon}
\orcid{0009-0004-4560-6361}
\affiliation{%
  \institution{LinkedIn}
  \city{Mountain View}
  \state{CA}
  \country{USA}
}
\email{lsimon@linkedin.com}

\author{Jingwei Wu}
\orcid{0009-0000-5751-9389}
\affiliation{%
  \institution{LinkedIn}
  \city{Mountain View}
  \state{CA}
  \country{USA}
}
\email{jingwu@linkedin.com}

\author{Wenjing Zhang}
\orcid{0009-0000-6510-1222}
\affiliation{%
  \institution{LinkedIn}
  \city{Mountain View}
  \state{CA}
  \country{USA}
}
\email{wzhang@linkedin.com}

\renewcommand{\shortauthors}{Zhang et al.}

\begin{abstract}
Modern industrial recommender systems must optimize across competing objectives, balancing semantic relevance with business metrics such as engagement and revenue. While bi-encoders dominate large-scale retrieval due to their efficiency, they collapse these heterogeneous signals into a single static embedding space. This design creates a fundamental limitation: once trained, the retriever cannot adapt to shifting objective priorities at serving time without retraining. Moreover, joint optimization with multi-objective losses often induces interference between objectives, leading to suboptimal trade-offs.
We propose \emph{Embedding Subspace Partitioning} (\emph{ESP}), a retrieval framework that decomposes the embedding into task-aware subspaces and replaces the single dot product with a weighted sum of per-subspace similarities, whose weights are tunable at serving time. For Transformer bi-encoders, ESP uses the model's native end-of-sequence token as a segment delimiter, with segment-aware attention masking and position encoding resets to guarantee subspace isolation in a single forward pass. Serving is performed via GPU-accelerated exhaustive kNN over one concatenated index, eliminating the need for per-objective Approximate Nearest Neighbor (ANN) infrastructure required by multi-head approaches.
We evaluate ESP on an open-source benchmark built from MS~MARCO~\cite{nguyen2016ms}. A single ESP model traces a broad Pareto frontier, consistently outperforming strong multi-task baselines across diverse operating points.
In LinkedIn's job matching platform (70M+ weekly users), ESP enabled dynamic retrieval reconfiguration and delivered significant key business metric lifts.
\end{abstract}



\maketitle

\section{Introduction}
\label{sec:intro}

Industrial retrievers rarely serve a single notion of relevance. For example, LinkedIn's job matching platform must produce a ranked candidate set that simultaneously captures semantic fit between members and jobs, predicts member engagement, and aligns with the marketplace's revenue objectives. Similar multi-objective trade-offs arise across recommendation domains, including e-commerce search (semantic relevance vs.\ brand and attribute matching) and news feeds (topical relevance vs.\ freshness and diversity). As a result, the first-stage retriever must generate candidates that balance multiple, and often competing, optimization objectives.

The prevailing paradigm for large-scale retrieval relies on bi-encoder architectures~\cite{karpukhin2020dense, ni2021sentence, huang2020embedding}, which project each item into a monolithic $d$-dimensional embedding space and quantify relevance via a single scalar dot product. This design underpins their efficiency, enabling Approximate Nearest Neighbor (ANN)~\cite{jegou2010product, huang2020embedding} indexing and low-latency serving at scale~\cite{li2021embedding, agarwal2024omnisearchsage}. However, the same design imposes a critical limitation in multi-objective settings: every signal must be packed into one geometry. Practitioners typically respond by (a)~training one embedding with a weighted multi-objective loss~\cite{zhang2022uni, guo2025gpu, li2021embedding}, which bakes the trade-off into training and forces a retrain on every priority change, or (b)~running a separate retrieval stack per objective and merging downstream, which fragments infrastructure linearly. We call this the \emph{no-knob} problem: the single-embedding bi-encoder gives operators no in-product dial between objectives such as semantic fit and revenue.

We posit that this is fundamentally architectural rather than algorithmic. Building on recent capacity bounds for bi-encoders~\cite{luo2025theoretical}, we show (Section~\ref{sec:theory}) that under conflicting objectives a shared embedding's per-objective capacity degrades with both the number and severity of conflicts; gradient-surgery methods such as PCGrad~\cite{yu2020gradient} mitigate but cannot remove the loss, because the bottleneck is representational, not optimizer-side.

We propose \textbf{Embedding Subspace Partitioning (ESP)}: partition the embedding into task-aware subspaces, each responsible for one objective, and replace the single dot product with a weighted sum of per-subspace dot products whose weights are tunable at serving time. For Transformer bi-encoders, we reuse the model's native end-of-sequence token as a segment delimiter, with entity-first ordering, segment-aware attention masking, and position-ID reset to guarantee subspace isolation in a single forward pass with no additional parameters; for DNN bi-encoders, ESP is realized via lightweight per-objective projection heads on a shared trunk.

At a glance this resembles multi-head architectures common in ranking, but applying multi-head to first-stage \emph{retrieval} is fundamentally harder: standard ANN indices such as IVFPQ~\cite{jegou2010product} cluster documents and approximate distances via residual quantization, which presupposes that the embedding is a coherent geometric object. A concatenation of $k$ subspace vectors is not: it is independent embeddings glued together, with no single neighborhood structure to cluster on. The alternative, $k$ independent ANN indices, scales infrastructure linearly with the number of objectives. ESP sidesteps both via GPU-accelerated exhaustive KNN~\cite{zhang2024scaling}: a model--hardware co-design that scores the weighted sum as a brute-force scan over one concatenated index, with weights applied on the request side as a serving-time knob. This approach has been deployed in LinkedIn's job matching platform for two years, serving 70M+ weekly active users; in a 7-day online A/B test, retuning the serving-time weights alone on a single trained model yielded +4.0\% revenue, +3.95\% job applications, and $-$8.95\% no-fit rate.

Our key contributions include the following:
\begin{itemize}
\item \textbf{A capacity argument} explaining why a single embedding's per-objective resolution degrades under conflicting objectives, motivating partitioning at the architecture level rather than at the loss level (Section~\ref{sec:theory}).
\item \textbf{ESP, an architecture-level realization of subspace partitioning} for both Transformer and DNN bi-encoders, scored by a single weighted dot product and served by GPU accelerated exhaustive KNN~\cite{zhang2024scaling}. ESP exposes the per-objective weights as a serving-time knob, letting operators trace the full Pareto frontier between objectives from a single trained model (Section~\ref{sec:method}).
\item \textbf{Empirical validation} on a multi-objective MS~MARCO~\cite{nguyen2016ms} benchmark with hard entity negatives (released open-source) and on LinkedIn's production system. ESP outperforms multi-loss and gradient-surgery baselines at every matched operating point and delivers the production gains above (Sections~\ref{sec:experiments} and~\ref{sec:production}).
\end{itemize}

\section{Related Work}
\label{sec:related}
\paragraph{Bi-Encoder Retrieval.}
Dense retrieval has evolved from early bi-encoder models such as DPR~\cite{karpukhin2020dense} and ANCE~\cite{xiong2020approximate} to more recent large-scale text encoders including E5~\cite{wang2024text} and GTE~\cite{li2023towards}. While these models achieve strong performance for single-objective retrieval, they fundamentally produce a single embedding per input. Existing approaches to multi-objective retrieval therefore either modify the training objective (discussed below) or alter the embedding architecture itself, which is the direction taken by ESP. Recent theoretical work by \citet{luo2025theoretical} establishes dimension-dependent capacity bounds for single-objective bi-encoders; in Section~\ref{sec:theory}, we extend this analysis to the multi-objective setting.
\paragraph{Multi-Objective Optimization in Retrieval.}
Prior work falls into three groups, all of which optimize \emph{training} while leaving the serving-time embedding architecture unchanged.

\emph{Multi-Task Single-Head (MTSH and other single-vector variants).} Uni-Retriever~\cite{zhang2022uni} and Taobao's system~\cite{li2021embedding} train one embedding under a weighted multi-objective loss; Baidu Search~\cite{guo2025gpu} compresses objectives into scalar scores before ANN; JRC~\cite{sheng2023joint} balances recall and calibration via dual member embeddings. Multi-Aspect Dense Retrieval (MADR)~\cite{kong2022madr} supervises per-aspect representations in intermediate layers but \emph{fuses them into a single output vector} via a learned aggregator, preserving ANN compatibility at the cost of  serving-time knob. These methods collapse the trade-off into training-time choices, produce a single served embedding, and suffer geometric interference that loss weighting cannot eliminate.

\emph{Multi-Task Multi-Head (MTMH).} MTMH~\cite{wang2025mtmh} adds $k$ task-specific projection heads on a shared encoder trunk, each producing a separate embedding. At serving time, each head requires its own ANN index with a post-retrieval merge, scaling infrastructure linearly with the number of objectives. ESP shares the partitioning principle but, as discussed in Section~\ref{sec:intro}, sidesteps the per-objective ANN problem via exhaustive KNN over one concatenated index~\cite{zhang2024scaling}, and adds Transformer-specific mechanisms.

\emph{Gradient conflict management.} PCGrad~\cite{yu2020gradient} projects conflicting task gradients during multi-task training, reducing interference but still producing a single embedding without a serving-time knob.
\paragraph{Other Multi-Task and Multi-Vector Approaches.}
Matryoshka Representation Learning (MRL)~\cite{kusupati2022matryoshka} partitions embeddings by \emph{resolution}---truncating trailing dimensions for efficiency---rather than by objective. Instruction-based methods such as LLM Embedder~\cite{zhang2023llmembedder} and Retrieval-GRPO~\cite{li2025retrievalgrpo} switch tasks via prompt or RL but produce one embedding per request, offering task selection rather than tunable trade-offs. OmniSearchSage~\cite{agarwal2024omnisearchsage} and Maillard et~al.~\cite{maillard2021multitask} apply multi-task losses without changing the embedding geometry. ColBERT~\cite{khattab2020colbert} achieves multi-aspect matching through per-token late interaction, at $O(L)$ vectors per document; ESP keeps storage at $k{=}2$--$3$ vectors, comparable to standard bi-encoders.


\section{Geometric Analysis of Multi-Objective Embedding}
\label{sec:theory}

We analyze the minimum embedding dimension required for a single bi-encoder to serve $k$ retrieval objectives simultaneously, and show that this minimum grows with a geometric notion of \emph{conflict} between objectives. 
While previous work~\cite{luo2025theoretical} focuses on the cardinality of distinct rankings in $\mathbb{R}^d$, we frame the multi-objective problem as one of subspace approximation. We derive a rigorous lower bound that the partitioned construction of ESP saturates in the worst case, providing an architectural guarantee of capacity.

\subsection{Problem Formulation}

\begin{definition}[Multi-Objective Bi-Encoder Retrieval]
\label{def:moo-bi}
Given a query set $\mathcal{Q}$, a document set $\mathcal{D}$, and $k$ relevance matrices $R_1, \ldots, R_k \in \mathbb{R}^{|\mathcal{Q}| \times |\mathcal{D}|}$, find encoders $f: \mathcal{Q} \to \mathbb{R}^d$ and $g: \mathcal{D} \to \mathbb{R}^d$ such that $S = f(\mathcal{Q})\,g(\mathcal{D})^\top$ approximates each $R_i$ within relative Frobenius error $\delta$:
\[
\|S - R_i\|_F \leq \delta \,\|R_i\|_F \qquad \forall i \in [k].
\]
\end{definition}

Since $\mathrm{rank}(S) \leq d$, the fundamental constraint is determining the smallest $d$ for which a single rank-$d$ matrix can simultaneously $\delta$-approximate the dominant subspaces of all $k$ objectives.

\subsection{Capacity Lower Bound for Shared Embeddings}
\label{sec:lower_bound}

For each $R_i = U_i \Sigma_i V_i^\top$, let the \emph{$\delta$-effective right subspace} $V_i^\delta \subseteq \mathbb{R}^{|\mathcal{D}|}$ be the span of the smallest set of right singular vectors capturing $1-\delta^2$ of the squared Frobenius norm. Let $r_i = \dim V_i^\delta$ denote the $\delta$-effective rank of objective $i$.

\begin{definition}[Subspace Conflict]
\label{def:conflict}
The pairwise conflict between objectives $i$ and $j$ is
\[
\varepsilon_{ij} \;:=\; 1 \;-\; \frac{\|P_{V_i^\delta}\,P_{V_j^\delta}\|_F^2}{\min(r_i, r_j)}\,,
\]
where $P_V$ is the orthogonal projector onto $V$ and $\|P_{V_i^\delta} P_{V_j^\delta}\|_F^2 = \sum_l \cos^2\theta_l$ is the sum of squared cosines of the principal angles $\{\theta_l\}$ between the two subspaces. Hence $\varepsilon_{ij} \in [0,1]$: $\varepsilon_{ij} = 0$ when one effective subspace is contained in the other (full alignment), and $\varepsilon_{ij} = 1$ when the subspaces are orthogonal.
\end{definition}

This metric captures conflict as an intrinsic representational property of the data geometry. It measures the principal-angle distance between the dominant ranking manifolds of two objectives, quantifying how much of one objective's signal already lies in the other's geometry, independent of optimizer or loss.

\begin{theorem}[Multi-Objective Capacity Lower Bound]
\label{thm:shared}
Any rank-$d$ score matrix $S$ satisfying $\|S - R_i\|_F \leq \delta\,\|R_i\|_F$ for all $i \in [k]$ must satisfy
\[
d \;\geq\; \dim\!\left(\sum_{i=1}^k V_i^\delta\right).
\]
Under the symmetric model where $r_i = r$ and $\varepsilon_{ij} = \varepsilon$ for all $i \neq j$, this simplifies to
\[
d \;\geq\; r \cdot \bigl(1 + (k-1)\varepsilon\bigr) \cdot \bigl(1 - O(\delta)\bigr).
\]
Equivalently, for a fixed budget $d$ the per-objective $\delta$-effective rank that can be jointly preserved is bounded above by
\[
r_{\mathrm{eff}}(d, k, \varepsilon) \;\leq\; \frac{d}{1 + (k-1)\varepsilon}\,.
\]
\end{theorem}

\begin{proof}[Proof sketch]
Let $P$ be the orthogonal projector onto $\mathrm{row}(S)^\perp$. Since $SP = 0$,
\[
\|R_i P\|_F \;=\; \|(R_i - S)\,P\|_F \;\leq\; \|R_i - S\|_F \;\leq\; \delta \,\|R_i\|_F\,,
\]
so the component of $R_i$ orthogonal to $\mathrm{row}(S)$ captures at most $\delta^2 \|R_i\|_F^2$. Combined with the definition of the $\delta$-effective subspace and Wedin's perturbation theorem~\cite{wedin1972}, this forces $V_i^\delta \subseteq \mathrm{row}(S)$ up to an $O(\delta)$ angle. Hence $\mathrm{row}(S) \supseteq \sum_i V_i^\delta$, so $d = \mathrm{rank}(S) \geq \dim(\sum_i V_i^\delta)$.

For the symmetric simplification, decompose each $V_i^\delta$ as a \emph{shared core} of dimension $r(1-\varepsilon)$ aligned across all objectives plus a \emph{private complement} of dimension $r\varepsilon$ orthogonal across objectives. Then $\dim(\sum_i V_i^\delta) = r(1-\varepsilon) + k\cdot r\varepsilon = r\bigl(1+(k-1)\varepsilon\bigr)$, and the bound on $r_{\mathrm{eff}}$ follows by inversion.
\end{proof}

This is the multi-objective analogue of the single-objective rank bound: representing $k$ potentially-conflicting objectives requires the embedding to span the union of their effective subspaces, which grows linearly in both $k$ and the geometric conflict $\varepsilon$.

\subsection{Constructive Upper Bound via Partitioning}

\begin{theorem}[Partitioning Saturates the Worst-Case Bound]
\label{thm:partitioned}
By partitioning the embedding into $k$ orthogonal blocks of sizes $\{d_i\}$ where $d_i \geq r_i$ and $\sum_i d_i = d$, and dedicating the $i$-th block to $R_i$, the resulting score matrix satisfies $\|S - R_i\|_F \leq \delta\,\|R_i\|_F$ for all $i$ at total dimension $d = \sum_i r_i$. Importantly, this is \emph{independent of the conflict pattern $\{\varepsilon_{ij}\}$}.
\end{theorem}

\begin{proof}
Each block is a rank-$d_i$ approximator of $R_i$ alone. By Eckart--Young~\cite{eckart1936}, $d_i \geq r_i$ guarantees $\delta$-accurate approximation within that block, and orthogonality of blocks ensures objectives do not interfere across the partition.
\end{proof}

\paragraph{Implication for ESP.}
The advantage of ESP is clear based on Theorems~\ref{thm:shared} and~\ref{thm:partitioned}: a shared embedding's required dimension $r(1+(k-1)\varepsilon)$ grows with the conflict $\varepsilon$, while partitioning's requirement $\sum_i r_i$ remains constant. The two bounds coincide only in the worst case ($\varepsilon \to 1$). At intermediate conflicts a shared embedding can in principle attain the smaller bound, but only if training successfully discovers and aligns the data's shared and private components. Partitioning trades this potential dimensional efficiency for an \emph{architectural guarantee}: the decomposition is built in, not learned. ESP's per-subspace structure (Section~\ref{sec:method}) is the direct realization of this guarantee, ensuring each objective receives a conflict-free slice of representation capacity, regardless of how the data's relevance subspaces overlap. Our experiments (Section~\ref{sec:experiments}) confirm that shared training in practice does not close the gap even with gradient-surgery methods such as PCGrad, leaving substantial capacity unused.

\subsection{Estimating Conflict in Practice}
\label{sec:conflict_estimate}

Directly computing $\varepsilon_{ij}$ requires the SVD of each $R_i$, which is computationally prohibitive at industrial scale. We use Kendall's $\tau$ between rankings induced by each objective on a held-out query sample as a tractable proxy:
\[
\hat{\varepsilon}_{ij} \;=\; 1 - \tau\bigl(\mathrm{rank}_i,\,\mathrm{rank}_j\bigr).
\]
Concretely, $\tau$ is computed as the average per-query Kendall correlation between rankings induced by $R_i$ and $R_j$ on the eval set's 1{,}359 queries over the $\sim$39K-passage corpus (Section~\ref{sec:retrieval_eval}). The proxy is \emph{directionally aligned} with subspace conflict: when subspaces overlap heavily, the dominant relevance directions agree and rankings track closely ($\tau \to 1$, $\hat\varepsilon \to 0$); when relevance subspaces orthogonalize, ranking correlations vanish ($\tau \to 0$, $\hat\varepsilon \to 1$). The relationship is not a strict equality: singular-value differences can induce ranking disagreement even at $\varepsilon = 0$. We treat $\hat\varepsilon$ as an order-of-magnitude indicator throughout the empirical sections.

\subsection{Optimal Dimension Allocation}
\label{sec:dim_alloc}

When the budget $d$ falls short of $\sum_i r_i$, partitioning cannot fully accommodate all objectives and dimensions must be allocated.

\begin{theorem}[Allocation by Effective Rank]
\label{thm:allocation}
Subject to $\sum_i d_i = d$ and $d_i \geq 1$, the allocation that maximizes the worst-case Frobenius coverage $\min_i \|P_{d_i}(R_i)\|_F^2 / \|R_i\|_F^2$ under uniform spectral decay across objectives is
\[
d_i^* \;=\; d \cdot \frac{r_i}{\sum_{j=1}^k r_j}\,.
\]
\end{theorem}

\begin{proof}[Proof sketch]
Truncating $R_i$ at rank $d_i$ leaves residual $\sum_{l > d_i} \sigma_l^2$. Under uniform spectral decay $\sigma_l^2 / \|R_i\|_F^2 \approx \tfrac{1}{r_i}\,\mathbb{1}[l \leq r_i]$, the residual is $1 - d_i / r_i$, and worst-case coverage is maximized when each objective's truncation level $d_i / r_i$ is equalized and gives the proportional allocation.
\end{proof}

\paragraph{Resolution gap in practice.}
Different objectives can have vastly different effective ranks. A \emph{low-resolution} objective such as categorical entity matching has small $r_i$: with $\sim 780$ unique entities, sphere-packing implies $r_i = O(\log 780) \approx 10$ dimensions suffice. A \emph{high-resolution} ranking objective has large $r_i$ scaling with corpus complexity. In our experiments (Section~\ref{sec:dim_sweep}), entity matching saturates at $d_{\mathrm{ent}} = 32$ (3\% of $H = 1024$) while semantic quality requires $d_{\mathrm{sem}} \gtrsim 750$. The resulting $\sim1{:}20$ resolution gap has an important systems implication: auxiliary business objectives can often be allocated compact subspaces, leaving most dimensions available for high-fidelity semantic retrieval. Objectives involving highly confusable attributes (e.g., near-duplicate product SKUs) require higher resolution and would correspondingly consume larger subspaces.

\section{Method: Embedding Subspace Partitioning}
\label{sec:method}

Motivated by the geometric analysis of multi-objective retrieval, we introduce Embedding Subspace Partitioning (ESP). ESP structures the embedding space into discrete, task-aware subspaces that enable controllable, per-objective scoring during inference. ESP is architecture-agnostic; we present two instantiations: (1)~for Transformer-based ESP, using the internal representational capacity of autoregression via specialized attention masking; (2)~for DNN-based ESP, using per-objective projection heads on a shared trunk. The variant is deployed in LinkedIn's production retrieval system.

The two instantiations are evaluated on different objective pairs: semantic vs.\ entity-attribute matching for Transformer-ESP (Section~\ref{sec:experiments}, on an open-source MS~MARCO benchmark~\cite{nguyen2016ms}) and engagement vs.\ revenue for DNN-ESP (Section~\ref{sec:production}, in production). This reflects the labels naturally available in each setting. Both are instances of the same multi-objective retrieval problem, and the partitioning principle is agnostic to which objectives are partitioned.

\subsection{ESP Architecture}
\label{sec:esp_arch}

ESP partitions the embedding into $k$ orthogonal subspaces, with the $i$-th subspace producing a sub-embedding $\mathbf{e}_i \in \mathbb{R}^{d_i}$ dedicated to the $i$-th objective. The full embedding is the concatenation $\mathbf{e} = [\mathbf{e}_1; \ldots; \mathbf{e}_k] \in \mathbb{R}^{d}$ with $d = \sum_i d_i$. Each $\mathbf{e}_i$ is L2-normalized so that $\langle \mathbf{e}_i, \mathbf{e}_i' \rangle \in [-1, 1]$; per-objective dimensions $\{d_i\}$ are allocated proportionally to per-objective effective rank, following Theorem~\ref{thm:allocation}.

\paragraph{Scoring with a serving-time knob.}
The retrieval score between query $\mathbf{q}$ and document $\mathbf{d}$ is a weighted sum of per-subspace dot products:
\begin{equation}
\label{eq:esp_score}
S(\mathbf{q}, \mathbf{d}) = \sum_{i=1}^{k} w_i \cdot \langle \mathbf{e}_i, \mathbf{e}_i' \rangle.
\end{equation}
The weights $\{w_i\}$ are applied on the request side at \emph{serving time}, not baked into the model: a single trained ESP model can dynamically trace the full Pareto frontier between objectives by varying $\{w_i\}$, with no retraining required. When $w_i = 1$ for all $i$, Equation~\ref{eq:esp_score} reduces to a standard dot product over the concatenated embedding $\mathbf{e}$, preserving full compatibility with existing exhaustive-KNN retrieval infrastructure~\cite{zhang2024scaling}.

\paragraph{Per-subspace training.}
Each subspace is trained with its own per-objective loss; the choice of loss family is unconstrained by the architecture. Contrastive losses (e.g.\ InfoNCE for ranking objectives), regression losses (e.g.\ MSE for revenue prediction), and classification losses (e.g.\ BCE for binary fitness) can be combined freely across subspaces. LinkedIn's production system (Section~\ref{sec:production}) uses all three. For ranking tasks, we utilize the InfoNCE loss:
\begin{equation}
\mathcal{L}_i = -\frac{1}{N} \sum_{n=1}^{N} \log \frac{\exp(\langle \mathbf{e}_i^{(n)}, \mathbf{e}_i'^{+(n)} \rangle / \tau)}{\sum_{j=1}^{B} \exp(\langle \mathbf{e}_i^{(n)}, \mathbf{e}_i'^{(j)} \rangle / \tau)},
\end{equation}
where $\tau$ is the temperature parameter and $B$ is the batch size. The total training loss is
\begin{equation}
\mathcal{L} = \sum_{i=1}^{k} \alpha_i \cdot \mathcal{L}_i.
\end{equation}
Because $\mathcal{L}_i$ depends only on $\mathbf{e}_i$, we have $\partial \mathcal{L}_i / \partial \mathbf{e}_j = 0$ for $j \neq i$: each loss updates only its own subspace's representation direction without gradient interference, so no ``compromise direction'' between objectives can arise at the output. The shared encoder parameters still receive mixed gradients from all $\mathcal{L}_i$, but the capacity guarantee of Theorem~\ref{thm:partitioned} rests on the output partition, not on the trunk; training-time weights $\{\alpha_i\}$ therefore affect convergence speed rather than geometric capacity.

\subsection{Realizing the Subspace Partition}
\label{sec:realizations}

The architecture above leaves open one engineering question: how are the per-subspace embeddings $\{\mathbf{e}_i\}$ computed for a given encoder type? We describe two realizations that produce the structurally orthogonal output partition required by Theorem~\ref{thm:partitioned}.

\subsubsection{Transformer Bi-Encoders: EOS-Delimited Segments}
\label{sec:esp_llm}

\paragraph{Input formatting.}
We use the model's \emph{native end-of-sequence token} (\emph{EOS}) as a segment delimiter, rather than adding new vocabulary tokens. This is critical: a pretrained transformer encoders already knows how to produce summary representations at EOS positions, so ESP works without fine-tuning. The attribute segment is placed before the main text:
\begin{align*}
&\langle\text{attribute value}\rangle~\texttt{<EOS>} \\
&\texttt{Instruct: ... query:}~ \langle\text{text}\rangle~\texttt{<EOS>}
\end{align*}
Documents follow an analogous format. Attribute-first ordering is preferred for two reasons. (i)~Under causal attention alone, the regime used in our pretrained model evaluation without encoder modification, the attribute EOS token naturally attends only to attribute tokens because no passage tokens precede it. This provides isolation for the attribute subspace without requiring additional masking. (ii)~When combined with the segment-aware mask below, the partition becomes \emph{symmetric}: both EOS embeddings depend only on their own segments, and Theorem~\ref{thm:partitioned}'s orthogonality assumption holds exactly.

The attribute value on both query and document sides must refer to the \emph{same typed field}, e.g.\ both contain the person name, or both contain the job title. This assumes an upstream entity tagger produces the attribute, which is standard in production retrieval systems. For queries without a tagged attribute, the attribute segment is omitted and ESP becomes standard single-embedding retrieval.

\paragraph{Segment-aware attention masking.}
\label{sec:attn_mask}
We apply a segment-aware causal attention mask in which each EOS token at position $p_i$ attends only to tokens in its own segment, i.e.\ positions $[p_{i-1}+1, \ldots, p_i]$. Both EOS embeddings are therefore restricted to their respective segments: the attribute EOS attends only to attribute tokens, while the semantic EOS attends only to passage tokens. This yields symmetric subspace isolation and satisfies the orthogonal block assumption of Theorem~\ref{thm:partitioned}. Layer normalization, which aggregates across attended positions, also operates within-segment under this mask, so cross-segment leakage cannot occur through any pathway. We use SDPA (not flash attention)~\cite{niu2021review, sanger2026scaled} to support custom 4D masks.

\paragraph{Position-ID reset (RoPE correction).}
\label{sec:rope}
Modern Transformers such as frontier LLMs use Rotary Position Embeddings (RoPE)~\cite{su2024roformer}, which makes hidden states position-dependent. Without correction, the same attribute text at different absolute positions (due to varying passage lengths) produces different embeddings. We fix this by resetting position IDs at each segment boundary: segment $i$ always starts at position 0 regardless of prior segment lengths. This ensures that identical attribute text on the query and document sides receives identical position encodings, producing cosine similarity $\approx 1$ as expected.

\paragraph{Embedding extraction.}
The hidden state $\mathbf{h}_i \in \mathbb{R}^H$ at each segment's EOS position is as the pre-normalization sub-embedding. If dimension allocation is applied, $\mathbf{h}_i$ is truncated to its first $d_i$ dimensions ($\sum_i d_i \leq H$) before L2-normalization to give $\mathbf{e}_i$. \emph{The Transformer realization introduces no additional parameters} beyond the base encoder; the DNN realization (Section~\ref{sec:esp_dnn}) adds $k$ lightweight projection heads.

\paragraph{Ablation summary.}
Entity-first ordering provides causal isolation for free: identical entity text yields cosine $\approx 0.97$ regardless of passage content. Adding the segment-aware mask further improves entity discriminability by $+$7.2pp (92.8\% $\to$ 100\%). Position reset corrects RoPE bias, restoring cosine from $\approx 0.50$ to $\approx 0.73$--$1.0$. We recommend causal-only ESP for pretrained evaluation, adding mask + position reset when fine-tuning.

\subsubsection{DNN Bi-Encoders: Projection Heads}
\label{sec:esp_dnn}

For DNN bi-encoders without causal attention, EOS-based segmentation does not apply; we instantiate the partition through $k$ trainable projection heads on a shared encoder trunk. Each head $h_i$ projects the trunk output into an objective-specific sub-embedding: $\mathbf{e}_i^q = h_i(f(q))$, $\mathbf{e}_i^d = h_i(g(d))$. Independent heads structurally enforce the output partition. LinkedIn's production retrieval system~\cite{zhang2024scaling} uses this design with $k=3$ objectives (engagement, revenue, fitness); see Section~\ref{sec:production} for results.

\section{Open-Source Experiments}
\label{sec:experiments}

We design our experiments to validate the theoretical framework established in Section~\ref{sec:theory}: we measure geometric conflict $\varepsilon$ from real-world data, predict capacity loss, train models with and without ESP, and demonstrate that ESP overcomes the interference ceiling predicted by Theorem~\ref{thm:shared}.

\subsection{Benchmark Construction}
\label{sec:benchmark}

Standard retrieval benchmarks often conflate semantic relevance with keyword matching. To isolate these objectives, we construct a multi-objective benchmark from MS~MARCO~\cite{nguyen2016ms} specifically designed with hard entity negatives.
\paragraph{Entity Extraction.} For each query, we extract multi-word named entities (e.g., ``Martin Luther King,'' ``New York City,'' ``machine learning'') that appear verbatim in the relevant passage. Generic phrases are filtered.
\paragraph{Hard Negatives.} For each query-entity pair, we sample passages from \emph{other queries} that also contain the entity string but answer a different question. For example, a query about ``Martin Luther King speeches'' gets a hard negative about ``Martin Luther King birthday'' --- same entity, different information need. This creates genuine conflict: a pure entity-matching signal would rank both equally, whereas a semantic signal must distinguish between them.
\paragraph{Dataset Statistics.} The benchmark contains 1,993 unique queries split into 13,951 training and 2,093 evaluation (query, positive, hard-negative) examples; the 2,093 eval examples cover 1,359 unique eval queries. Of the eval pairs, \textbf{72\% are hard negatives} (both positive and negative contain the entity) and 28\% are easy negatives (entity absent from negative). All retrieval results in Section~\ref{sec:retrieval_eval} are reported on the 1,359 eval queries. The entity-first input format (Section~\ref{sec:esp_llm}) is used for all methods.

\subsection{Empirical Measurement of Conflict}
\label{sec:conflict_measure}

We first quantify the objective conflict using a pretrained open-source 0.6B embedding model ($H{=}1024$)~\cite{qwen3embedding}. By Kendall's $\tau$ between the semantic ranking (which passage is semantically relevant?) and the entity ranking (which passage contains the entity?), we find:

\[
\tau = 0.34, \quad \varepsilon = 1 - \tau = 0.66
\]

Applying Theorem~\ref{thm:shared} for $k{=}2$ objectives, we predict a maximum effective rank of $r_{\text{eff}} \leq 1024 / (1.66) \approx 616$. This represents a 40\% reduction in available representational capacity for any shared-embedding model attempting to satisfy both objectives simultaneously. 
This is the worst-case theoretical limit. The next two subsections verify the prediction empirically: training (Section~\ref{sec:retrieval_eval}) realizes a smaller but monotonic capacity loss in shared embeddings that no choice of loss weighting or gradient surgery can eliminate, while ESP's partitioned architecture avoids the loss entirely.

\subsection{Retrieval Evaluation}
\label{sec:retrieval_eval}

We evaluate all methods using Semantic NDCG@10 (MS~MARCO relevance) and Entity Recall@10 (fraction of entity-matched passages in top-10) over a corpus of $\sim$39K passages (eval set plus 35K randomly sampled MS~MARCO distractors). For each of 1{,}359 queries, all methods rank the full corpus.
\paragraph{Pretrained baselines (Table~\ref{table:retrieval}).} We compare four retrieval architectures on the pretrained 0.6B model~\cite{qwen3embedding}. Because pretrained MTMH lacks trained projection heads, we approximate per-head specialization by issuing separate prompts for each objective to the same encoder, retrieving the top-$n/2$ results from each corresponding index, and merging the candidate sets downstream.

\begin{table}[tb]
\caption{Pretrained retrieval (0.6B, 39K corpus, 1{,}359 queries). Sem: NDCG@10. Ent: Entity Recall@10. Storage: number of vectors per document.}
\small
\begin{tabular}{lcccc}
\toprule
\textbf{Method} & \textbf{Sem} & \textbf{Ent} & \textbf{Vectors / Doc} & \textbf{Scoring} \\
\midrule
Baseline (no entity) & .668 & .496 & 1 & dot \\
Entity in prompt & .762 & .594 & 1 & dot \\
MTMH${}^{\dagger}$~\cite{wang2025mtmh} & .677 & .833 & 2 & 2$\times$ dot, merge \\
ColBERT~\cite{khattab2020colbert} & .849 & .935 & $L$ token vectors & MaxSim \\
\textbf{ESP (ours)} & \textbf{.843} & \textbf{.932} & 1 & \textbf{dot} \\
\bottomrule
\multicolumn{5}{l}{\small ${}^{\dagger}$Pretrained MTMH simulated via per-objective prompts.}
\end{tabular}
\label{table:retrieval}
\end{table}

Without entity information, the baseline achieves .668 semantic NDCG. Adding entity text to the input via instruction-based routing~\cite{zhang2023llmembedder} (prepending entity to both query and passage) improves to .762/$+$14pp. However this is \emph{task switching}, not multi-objective optimization: the single embedding entangles entity and semantic signals with no mechanism for trade-off. MTMH's $n/2$ merge splits the retrieval budget between objectives, degrading both metrics. 
ColBERT achieves the strongest overall quality, but at substantially higher serving cost. It stores $O(L)$ token vectors per document and performs per token MaxSim interactions instead of a single vector similarity computation, where $L$ denotes document length.
In contrast, ESP achieves entity recall comparable to ColBERT at standard bi-encoder cost: it stores a single concatenated embedding per document composed of $k$ subspaces and computes relevance using one weighted similarity function, served by GPU-accelerated exhaustive KNN~\cite{zhang2024scaling}.
ESP's combined Sem NDCG (.843) exceeds entity-in-prompt (.762) because ESP's subspace isolation prevents entity signal from interfering with semantic ranking.
\paragraph{Trained interference ceiling (Table~\ref{table:trained}).} Trained ColBERT is omitted because its $\sim L$-vectors-per-document storage ($\sim 350$TB) precludes deployment in our target setting; the trained comparison focuses on single-vector and partitioned architectures. Increasing entity emphasis in a single embedding (MTSH) progressively degrades semantic quality while entity recall saturates. PCGrad~\cite{yu2020gradient} gradient surgery \emph{mitigates} the interference ($-4.6$pp semantic vs.\ $-4.9$pp for MTSH $\alpha_e{=}5$) but cannot eliminate it, confirming the bottleneck is \emph{geometric} (embedding capacity), not gradient conflict. ESP breaks through: at $w_e{=}0.8$, it achieves .900 semantic \emph{and} .999 entity --- outperforming the best MTSH at matched entity recall by $+$2.5pp semantic NDCG ($p < 0.001$, bootstrap 95\% CI [$+$1.6, $+$3.5] over 1{,}359 queries), without requiring retraining.

\begin{table}[tb]
\caption{Trained models (0.6B, 1 epoch, 39K corpus, 1{,}359 queries). MTSH: multi-task single-head with entity weight $\alpha_e$. ESP weight $w_e$ is tuned at serving time (no retraining). $\pm$ indicates std across queries.}
\small
\begin{tabular}{lcc}
\toprule
\textbf{Method} & \textbf{Sem NDCG@10} & \textbf{Ent Recall@10} \\
\midrule
Baseline (no entity) & .870 $\pm$ .220 & .679 \\
MTSH ($\alpha_e{=}1$) & \textbf{.924} $\pm$ .164 & .887 \\
MTSH ($\alpha_e{=}5$) & .875 $\pm$ .208 & 1.000 \\
MTSH ($\alpha_e{=}10$) & .823 $\pm$ .240 & 1.000 \\
MTSH + PCGrad~\cite{yu2020gradient} & .879 $\pm$ .205 & 1.000 \\
\midrule
\textbf{ESP} ($w_e{=}0.3$) & .899 $\pm$ .188 & .908 \\
\textbf{ESP} ($w_e{=}0.5$) & .900 $\pm$ .183 & .986 \\
\textbf{ESP} ($w_e{=}0.8$) & \textbf{.900} $\pm$ .183 & \textbf{.999} \\
\bottomrule
\end{tabular}
\label{table:trained}
\end{table}

\subsection{Pareto Frontier: A Single-Model Solution}
\label{sec:weight_sweep}

Beyond capacity loss, shared embeddings suffer a practical limitation: they provide \textbf{no mechanism to control the trade-off} between objectives at serving time, and any emphasis must be baked into training loss weights. ESP resolves this with serving-time control. Table~\ref{table:sweep} compares ESP's weight sweep (one trained model) against MTSH's Pareto frontier (each point requires a separate training run).

\begin{table}[tb]
\caption{Pareto comparison (1{,}359 queries): ESP (one model, vary $w_e$ at serving time) vs.\ MTSH (retrain for each $\alpha_e$). Rows ordered by increasing entity emphasis.}
\small
\begin{tabular}{cc|cc}
\toprule
\multicolumn{2}{c|}{\textbf{ESP (one model)}} & \multicolumn{2}{c}{\textbf{MTSH (retrain each)}} \\
Sem & Ent & Sem & Ent \\
\midrule
.870 ($w_e{=}0$) & .693 & .870${}^{\dagger}$ & .679 \\
.899 ($w_e{=}0.3$) & .908 & .924 ($\alpha_e{=}1$) & .887 \\
.900 ($w_e{=}0.5$) & .986 & .875 ($\alpha_e{=}5$) & 1.000 \\
.900 ($w_e{=}0.8$) & .999 & .823 ($\alpha_e{=}10$) & 1.000 \\
\bottomrule
\multicolumn{4}{l}{\small ${}^{\dagger}$Baseline model (no entity in training).}
\end{tabular}
\label{table:sweep}
\end{table}

ESP traces a smooth Pareto frontier using a single trained model. At comparable levels of entity recall, ESP consistently matches or exceeds the NDCG achieved by MTSH, while MTSH requires retraining for each operating point. The value of ESP lies not in any single operating point but in the ability to trace the  Pareto frontier from a single trained model, a capability MTSH  lacks.
\paragraph{Semantic NDCG plateau and signal complementarity.} ESP's semantic NDCG is notably stable across $w_e \in [0.3, 0.8]$ (.899 to .900). This plateau reflects \emph{signal complementarity}, not absence of conflict. At $w_e{=}0$ (semantic only), NDCG is .870; adding entity weight to $w_e{=}0.3$ \emph{improves} semantic NDCG by +2.9pp, because entity-matched passages in our benchmark are correlated with semantic relevance. Within the operating range $[0.3, 0.8]$, the passages promoted by increasing $w_e$ are entity-matched \emph{and} semantically relevant, so semantic quality is maintained. At the extreme ($w_e{=}1.0$), semantic NDCG collapses to .583, confirming the trade-off exists but only manifests when entity signal overwhelms semantics entirely. This plateau is precisely what subspace isolation enables: because semantic and entity scores are computed independently, entity emphasis cannot \emph{corrupt} the semantic ranking --- it can only \emph{reweight} the combined score. In MTSH, where both signals share one embedding, increasing entity emphasis degrades semantic quality monotonically (0.924 $\to$ 0.823) because the signals interfere geometrically.

\subsection{Dimension Allocation}
\label{sec:dim_sweep}

We sweep $d_{\text{ent}}$ from 0 to $H$ by truncating the hidden state at each EOS, keeping $d_{\text{sem}} = H - d_{\text{ent}}$. With entity-first format providing perfect subspace isolation (cosine $\approx 1.0$ for identical entities):
\paragraph{Entity matching requires trivial dimensionality.} Both pretrained and trained models achieve 100\% entity accuracy at just $d_{\text{ent}} = 32$ (3\% of $H{=}1024$), matching the resolution-gap argument of Section~\ref{sec:dim_alloc}: with $\sim$780 unique entities, $O(\log 780) \approx 10$ bits are needed --- well within 32 dimensions.

\begin{figure}[tb]
\centering
\includegraphics[width=\columnwidth]{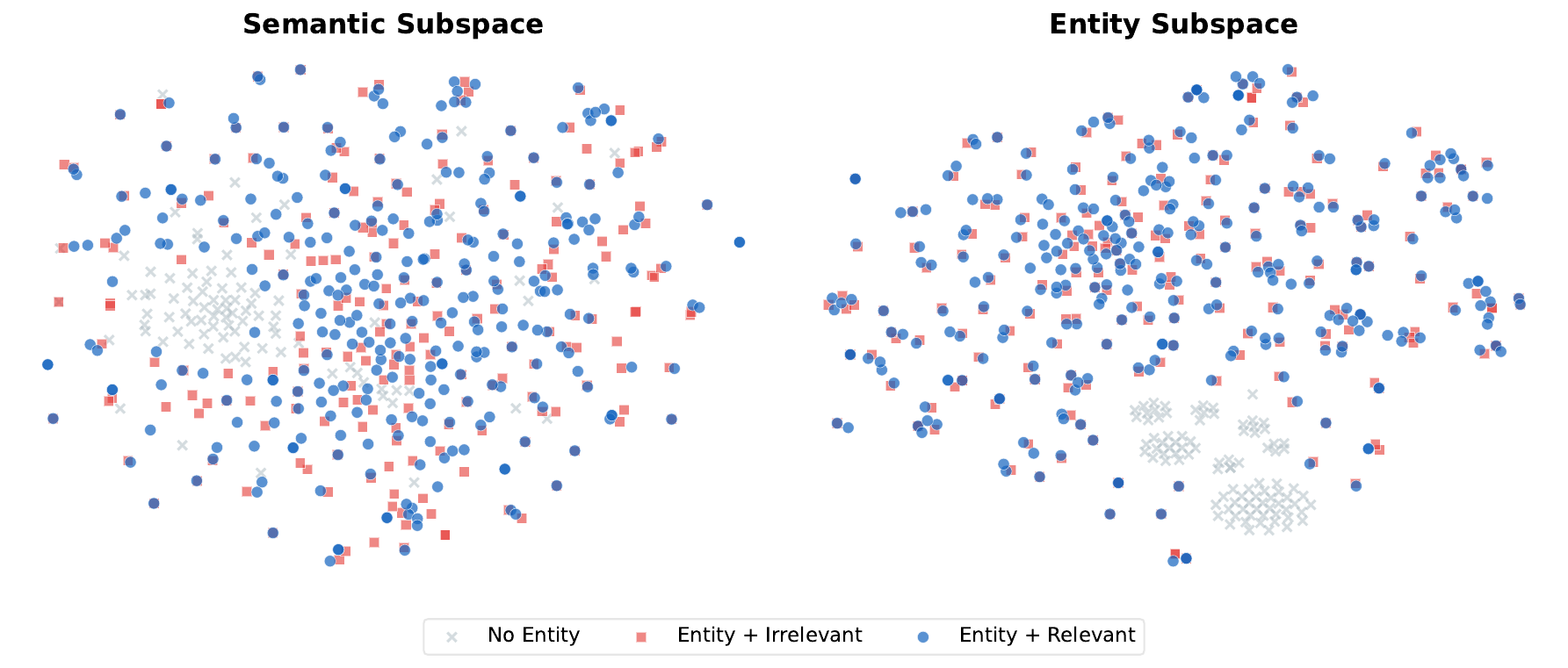}
\caption{t-SNE of ESP subspace embeddings. \textbf{Left}: Semantic subspace --- entity presence is entangled with content. \textbf{Right}: Entity subspace --- same-entity passages cluster tightly regardless of relevance.}
 \Description{Two side-by-side t-SNE scatter plots of passage
  embeddings. In the left plot, passages associated with the same
  entity are mixed with passages representing different semantic
  content. In the right plot, passages associated with the same
  entity form compact and clearly separated clusters, showing that
  entity identity is isolated in the entity subspace.}
\label{fig:tsne}
\end{figure}

Figure~\ref{fig:tsne} visualizes the subspace separation: entity identity is entangled with semantic content in the shared embedding (left), but cleanly separable in its own subspace (right), confirming entity matching is a low-dimensional signal trivial to isolate.

\subsection{Single-Objective Generality}
\label{sec:beir}

A multi-objective architecture is only useful if it does not regress on single-objective retrieval. We verify ESP's behavior in this regime by evaluating the pretrained model (with segment mask and position reset) on three BEIR benchmarks~\cite{thakur2021augmented} using a ``none'' entity prefix---reducing ESP to a standard bi-encoder. ESP shows $\leq$0.3pp variation vs.\ baseline across NFCorpus, SciFact, and FiQA, indicating no meaningful regression: the entity subspace embedding becomes constant across all documents and does not influence ranking.

\subsection{Scaling to Three Objectives
  (\texorpdfstring{$k=3$}{k=3})}
\label{sec:k3}

To test ESP's generality beyond $k{=}2$, we add a third objective: \emph{passage length preference}. We bucket passages into short ($<$30 words), medium (30--80), and long ($>$80), and set the query preference to ``medium'' (the most informative length). The input format becomes: \texttt{<length><EOS><entity><EOS><passage><EOS>}, producing three subspace embeddings scored by $S = w_s \cdot s_{\text{sem}} + w_e \cdot s_{\text{ent}} + w_l \cdot s_{\text{len}}$.

\begin{table}[tb]
\caption{$k{=}3$ ESP: semantic + entity + length preference (pretrained 0.6B, 1{,}359 queries). \textbf{Sem}: Semantic NDCG@10. \textbf{Ent}: Entity Recall@10. \textbf{Len}: Length Match@10 (fraction of top-10 matching preferred length).}
\small
\begin{tabular}{lccc}
\toprule
\textbf{Weights} $(w_s, w_e, w_l)$ & \textbf{Sem} & \textbf{Ent} & \textbf{Len} \\
\midrule
(1.0, 0.0, 0.0) sem only & .867 & .687 & .812 \\
(0.5, 0.5, 0.0) $k{=}2$ & .861 & .985 & .858 \\
\midrule
(0.5, 0.3, 0.2) $k{=}3$ balanced & .843 & .957 & .912 \\
(0.4, 0.4, 0.2) $k{=}3$ ent+len & .824 & .969 & .912 \\
(0.5, 0.2, 0.3) $k{=}3$ heavy len & .830 & .918 & .931 \\
\bottomrule
\end{tabular}
\label{table:k3}
\end{table}

We use length preference as a \emph{toy third objective} to demonstrate that ESP's serving-time control extends to $k > 2$; the production system in Section~\ref{sec:production} validates $k{=}3$ on three high-stakes objectives (engagement, revenue, fitness). At $(0.5, 0.3, 0.2)$, the model achieves .843 semantic, .957 entity, and .912 length match to simultaneously control three objectives from a single pretrained model with no additional training. Compared to $k{=}2$ at $(0.5, 0.5, 0)$, adding the length objective costs $-$1.8pp semantic and $-$2.8pp entity while gaining $+$5.4pp length match. This is a smooth, predictable trade-off rather than catastrophic interference.

\section{Production Validation at LinkedIn}
\label{sec:production}


While the preceding experiments on Transformer-based architectures validate the geometric principles and capacity benefits of ESP, we now demonstrate the framework's industrial viability and scalability. We deployed ESP within LinkedIn’s job matching platform, a high-throughput system serving over 70M weekly active users. It utilizes the DNN-ESP variant, employing $k{=}3$ dedicated projection heads branching from a shared encoder trunk. By transitioning from the controlled benchmark to a live production environment, we validate that the partitioning principle remains robust across different architectures and delivers business impact at scale.

\subsection{Objectives and Labels}

In our production system, we jointly optimize 3 objectives.

\paragraph{Engagement.} We utilize the user interaction data (135M clicks, 12M applies, 3M saves over 30 days) with in-batch negative sampling~\cite{Hidasi2106} and random easy negatives~\cite{yang2020mixed}. The loss is InfoNCE with temperature $T=0.01$.
\paragraph{Revenue.} We estimate revenue as $P(\text{click} | \text{impression}) \times \text{bid}$ under CPC pricing. We leverage knowledge distillation~\cite{gou2021knowledge} from a high-precision L2 ranker to generate soft $\text{pClick}$ labels via counterfactual data~\cite{bottou2013counterfactual}. Training uses MSE loss against these pseudo-labels.
\paragraph{Fitness.} We utilize soft labels generated from an LLM fine-tuned on 25K human annotations, quantifying how well a job matches a seeker's qualifications. Binary labels are derived by thresholding, and training uses BCE loss.

\subsection{Offline Results: ESP vs.\ MTSH}

Table~\ref{table:prod_moo} compares architectures for jointly optimizing engagement and revenue. MTSH uses a single embedding with a weighted training loss ($\alpha$: revenue loss weight). JRC~\cite{sheng2023joint} is a single-vector baseline that learns dual member embeddings for joint recall and calibration; like MTSH it cannot retune the engagement/revenue trade-off at serving time. ESP uses partitioned subspaces with per-objective serving-time weights ($w_{\text{rev}}, w_{\text{eng}}$).

MTSH shows severe interference at scale: emphasizing revenue ($\alpha{=}0.99$) loses 28.24\% in-batch recall and 21.19\% offline KNN recall. \emph{At matched revenue emphasis} (MTSH $\alpha{=}0.99$ vs.\ ESP $w_{\text{rev}}{=}0.99$), MTSH loses 28.24\% in-batch recall while ESP loses only 0.54\%, a $\approx 50{\times}$ improvement that quantifies the interference that partitioning eliminates at production scale. JRC mitigates the recall loss ($-0.03\%$) but inherits MTSH's lack of serving-time tunability and pays a large MSE penalty ($+5307.25\%$).

\begin{table}[tb]
\centering
\caption{Production multi-objective co-training results. MTSH: single embedding, weighted training loss ($\alpha$: revenue weight). ESP: partitioned embedding with per-objective serving weights ($w_{\text{rev}}, w_{\text{eng}}$). \emph{MSE diffs} are relative to the revenue-only baseline (MTSH $\alpha{=}0.99$); \emph{recall diffs} are relative to the engagement-only baseline (MTSH $\alpha{=}0.01$). 
}
\setlength{\tabcolsep}{3pt}
\begin{tabular}{@{}lccc@{}}
\toprule
\textbf{Model} & \textbf{\shortstack{MSE diff\\(\%)}} & \textbf{\shortstack{in-batch-2048\\@10 diff (\%)}} & \textbf{\shortstack{offline-KNN\\@6400 diff (\%)}} \\
\midrule
MTSH ($\alpha{=}0.01$) & +6905.80 & 0.00 & 0.00 \\
MTSH ($\alpha{=}0.99$) & 0.00 & $-$28.24 & $-$21.19 \\ \hline
JRC~\cite{sheng2023joint} & +5307.25 & $-$0.03 & $-$0.45 \\ \hline
ESP $(0.5, 0.5)$ & +52.17 & $-$1.10 & $-$1.29 \\
ESP $(0.8, 0.2)$ & +42.03 & $-$0.67 & $-$1.21 \\
ESP $(0.9, 0.1)$ & +23.19 & $-$0.84 & $-$0.76 \\
ESP $(0.99, 0.01)$ & +4.35 & $-$0.54 & $-$0.17 \\
\bottomrule
\end{tabular}
\label{table:prod_moo}
\end{table}

\paragraph{Scaling to $k{=}3$ in production.}
We add a fitness subspace alongside engagement and revenue. The fitness subspace achieves AUC 0.938 on held-out fitness annotations, while engagement and revenue metrics remain within parity of the $k{=}2$ ESP configuration above. This validates that ESP's per-subspace partitioning scales to three high-stakes objectives without measurable cross-objective regression, and provides the production counterpart to the open-source $k{=}3$ proof-of-concept in Section~\ref{sec:k3}.

\subsection{Online A/B Test}

We deployed the ESP model trained with $\alpha=0.99$ and varied the serving-time weight $w_{d_1}$ for revenue, with $w_{d_2}=1-w_{d_1}$ controlling engagement. Table~\ref{table:online_ab} reports statistically significant results ($p < 0.05$) from a 7-day online A/B test.

\begin{table}[tb]
\caption{Online A/B Test ($d_1$: revenue, $d_2$: engagement)}
\begin{tabular}{lccc}
\toprule
\textbf{Weight $w_{d_1}$} & \textbf{Revenue} & \textbf{Job Application} & \textbf{Click Rate} \\
\midrule
0.1 & +0.91\% & +3.30\% & +0.04\% \\
0.2 & +2.31\% & +3.20\% & +0.37\% \\
0.3 & +4.01\% & +3.95\% & +0.92\% \\
0.5 & +7.36\% & +1.88\% & +0.61\% \\
\bottomrule
\end{tabular}
\label{table:online_ab}
\end{table}

At $w_{d_1}=0.3$ (the chosen production operating point), the system achieves simultaneous gains across all metrics; post-ramp analysis at this configuration shows an 8.95\% reduction in the no-fit rate --- the rate at which recommended jobs are flagged as poorly matched to seeker qualifications by downstream filters and human review. \emph{A single trained ESP model spans four production operating points without retraining}, with revenue gains ranging from $+0.91\%$ to $+7.36\%$ and complementary engagement effects. The serving-time knob translates directly into deployable business levers, allowing product teams to rebalance objectives in response to evolving priorities without model retraining or infrastructure changes.

\section{Discussion and Future Work}
\label{sec:discussion}
\paragraph{Serving Latency (Transformer-ESP)} 
For the Transformer-ESP variant evaluated on the open-source benchmark (0.6B encoder on H200 GPUs), ESP introduces zero latency overhead: the causal-only variant (flash attention~\cite{dao2022flashattention}) adds $<$0.1\% over standard encoding, and the masked variant (SDPA)~\cite{niu2021review, sanger2026scaled} is 12\% \emph{faster} than flash attention on short sequences (22.4ms vs.\ 19.6ms per document). The custom 4D mask adds negligible cost beyond the SDPA switch itself. The production DNN-ESP system is much smaller and uses a different inference stack; we report production retrieval-quality results in Section~\ref{sec:production} and refer readers to~\cite{zhang2024scaling} for system-level latency.
\paragraph{Attribute Tagging as a Separate Problem}
ESP decomposes multi-objective retrieval into two sub-problems: (1)~entity/attribute tagging (upstream), and (2)~partitioned embedding matching (ESP). The tagging model extracts the same typed attribute on both query and document sides. This is standard practice in production: job search systems extract job titles, e-commerce systems extract brand names. It is not a limitation of ESP but rather a natural pipeline decomposition. For queries without tagged attributes, the attribute segment is omitted and ESP gracefully reduces to standard single-embedding retrieval.
\paragraph{Infrastructure Prerequisite: Exhaustive KNN}
ESP's weighted dot product $\sum_i w_i \langle \mathbf{e}_i, \mathbf{e}_i' \rangle$ is incompatible with partition-based ANN indices such as IVFPQ~\cite{jegou2010product}, whose clustering and residual quantization presuppose a single coherent embedding geometry. A concatenation of independently partitioned subspaces does not satisfy this assumption. Recent advances in GPU accelerated exhaustive retrieval~\cite{zhang2024scaling, malkov2018efficient} make ESP practical at web scale by enabling exact serving time reweighting over a unified index. As a result, ESP should be viewed as a joint model and systems design enabled by modern computation hardware.
%


\paragraph{Future Work}

ESP opens several promising directions for future research in multi-objective retrieval. First, an important next step is evaluating settings where multiple objectives simultaneously require high-resolution embeddings, such as balancing semantic precision against topical diversity. Developing benchmarks with such high-entropy objective conflicts would provide a stronger test of the representational limits estimated by our theory.
Second, we plan to improve the calibration of per-subspace scores using techniques such as post-hoc scaling or isotonic regression~\cite{barlow1972isotonic}, which could simplify control over Pareto-optimal operating points as the number of objectives $k$ grows. Extending ESP to larger numbers of heterogeneous objectives also remains an important practical direction.
Finally, several theoretical and systems challenges remain open, including tightening the dimension allocation bounds in Theorem~\ref{thm:allocation}, improving cross-objective gradient flow in the shared encoder trunk, and developing more efficient segment-aware attention masking compatible with FlashAttention-style~\cite{dao2022flashattention} kernels. We also plan to continue improving the benchmark infrastructure introduced in this work to support more rigorous and reproducible research on multi-objective retrieval systems.

\section{Conclusion}
\label{sec:conclusion}

We formalized the geometric limitation of single-embedding multi-objective bi-encoder retrieval and proposed ESP. The core practical insight is the \emph{no-knob} problem: a single embedding provides no mechanism to control the trade-off between competing objectives at serving time, and the interference ceiling is architectural, not a training artifact. ESP resolves this by decomposing the embedding into isolated subspaces with a tunable weighted dot product. Unlike multi-head retrieval architectures that require per-objective ANN indices and a post-hoc merge, ESP serves via a single concatenated index and a single weighted dot product, enabled by GPU-accelerated exhaustive KNN~\cite{zhang2024scaling}. In production at LinkedIn's job matching platform (70M+ weekly active users), ESP yields $+$4.0\% revenue, $+$3.95\% job applications, and $-$8.95\% no-fit rate from a single trained model spanning multiple operating points without retraining.
%


\begin{acks}
We thank the many talented engineers across LinkedIn whose dedication made this work possible, and our Product and Engineering leadership for their support.  
\end{acks}

\balance

\bibliographystyle{ACM-Reference-Format}
\bibliography{esp_references}

\end{document}